\documentclass[conference,9pt]{IEEEtran}

\usepackage[cmex10]{amsmath}

\usepackage{subfig}

\usepackage{epsfig,amsmath,amssymb,epsf,cite,algorithm,algorithmic}
\usepackage{graphicx}

\newtheorem{theorem}{Theorem}
\newtheorem{lemma}{Lemma}

\usepackage{stmaryrd}
\usepackage{amssymb}
\usepackage{tipa}
\usepackage{setspace}

\usepackage[none]{hyphenat}

\usepackage{verbatim,color,paralist}
\newcommand{\DG}[1]{{\color{blue}#1}}

\def\b0{{\pmb{0}}}

\def\ba{{\boldsymbol{a}}}   
   
  \def\bq{{\boldsymbol{q}}} \def\bw{{\boldsymbol{w}}}
   \def\bx{{\boldsymbol{x}}}
   \def\by{{\boldsymbol{y}}}
   \def\bz{{\boldsymbol{z}}}
\def\bg{{\boldsymbol{g}}} 

\def\bA{{\boldsymbol{A}}} \def\bH{{\boldsymbol{H}}}

\def\bG{{\boldsymbol{G}}}

\newenvironment{proof}[1][Proof]{\begin{trivlist}
\item[\hskip \labelsep {\bfseries #1}]}{\end{trivlist}}

\IEEEoverridecommandlockouts

\begin{document}
%
\title{A GENERALIZED LDPC FRAMEWORK FOR ROBUST AND SUBLINEAR COMPRESSIVE SENSING}
\author{Xu~Chen and Dongning~Guo \\
Department of Electrical Engineering and Computer Science\\
 Northwestern University, Evanston, IL, 60208, USA
\thanks{This work was supported in part by the National Science Foundation under Grant No. CCF-1423040.}
}

\maketitle

\begin{abstract}
Compressive sensing aims to recover a high-dimensional sparse signal from a relatively small number of measurements. In this paper, a novel design of the measurement matrix is proposed. The design is inspired by the construction of generalized low-density parity-check codes, where the capacity-achieving point-to-point codes serve as subcodes to robustly estimate the signal support. In the case that each entry of the $n$-dimensional $k$-sparse signal lies in a known discrete alphabet, the proposed scheme requires only $O(k \log n)$ measurements and arithmetic operations. In the case of arbitrary, possibly continuous alphabet, an error propagation graph is proposed to characterize the residual estimation error. With $O(k \log^2 n)$ measurements and computational complexity, the reconstruction error can be made arbitrarily small with high probability.
\end{abstract}

\section{Introduction}
\label{sec:intro}


Compressive sensing aims to recover a high-dimensional sparse signal from a relatively small number of measurements~\cite{donoho2006compressed,candes2006near}. There are two different designs of the measurement matrices: random construction and deterministic construction. Convex optimization approaches have been first proposed to recover the noiseless signal with $O(k \log (n/k))$ random measurements~\cite{candes2005decoding}. Greedy algorithms that involve lower complexity have been proposed~\cite{tropp2007signal,needell2009cosamp,donoho2012sparse}. However, most of the algorithms that are based on random measurement matrix design inevitably involve a complexity of ${\rm poly}(n)$.

Inspired by the error control code designs, deterministic structured measurement matrices have been proposed to reduce the computational complexity to (near) linear time $O(n)$~\cite{xu2007efficient,indyk2008near}. In practice, when the signal dimension is many thousands or millions, even linear time complexity often becomes prohibitive. In response, sublinear compressive sensing based on second order Reed-Muller codes has been proposed, but the reconstruction error was not characterized~\cite{applebaum2009chirp,calderbank2010construction}.

Recently, compressive sensing schemes with a novel design of measurement matrix and sublinear recovery algorithms have been developed, requiring $O(k)$ measurements and arithmetic operations under the noiseless setting~\cite{pawar2012hybrid,bakshi2012sho}. In those schemes, the measurements are split into multiple groups and each group is a sub-vector, which are linear combinations of the {\em same} set of signal components. Treating the measurement groups as bins, the design matrix basically hashes the signals to different measurement bins, which is similar to the bipartite graph induced by low-density parity-check~(LDPC) code structure. In \cite{pawar2012hybrid} and \cite{bakshi2012sho}, the measurement vector in each bin is designed to carry the signal support information by leveraging the discrete Fourier transform (DFT) matrix. The design has been extended to the noisy case, involving $O \left( k \log^{1.3} n \right)$ measurements and computational complexity, with the limitation that the signal entries must lie in a known discrete alphabet.

In this paper, we propose a generalized LDPC code inspired compressive sensing scheme to further reduce the the number of measurements required and computation complexity. The scheme adopts the sublinear recovery algorithm framework in \cite{pawar2012hybrid}. For the measurement matrix design, the scheme also adopts the LDPC structure to disperse the signal into measurement bins. The main difference is that each measurement bin is a subcode, where some recently developed capacity-achieving codes are utilized to encode the signal support.

Specifically, this paper makes the following contributions. First, our scheme is the first to achieve nearly order optimal $O(k \log n)$ noisy measurements and computational complexity for the case of known discrete alphabet. Second, the previous design based on DFT matrix is susceptible to quantization errors, while the proposed measurement matrix consists of only $\{0,\pm 1\}$ entries, which are easier and more robust in practice. Third, we propose an error propagation graph with error message passing rules to capture the error propagation for the case of arbitrary signals with unknown alphabet. Analysis shows that with $O \left( k \log^2 n \right)$ measurements and complexity the signal estimation error can be made arbitrarily small as $k$ increase. The proposed design and error propagation graph have potential applications in sparse Fourier transform~\cite{pawar2013computing} and Walsh-Hadamard transform with arbitrary signal alphabet~\cite{chen2015robust}.

\section{System Model}

Suppose $\bx \in {\mathbb{R}}^n$ is a $k$-sparse vector. The problem is to recover $\bx$ from the $m$-dimensional ($m \ll n$) measurement vector
\begin{align}\label{eq:system}
\by = \bA \bx + \bz
\end{align}
where $\bA \in \mathbb{R}^{m \times n}$ is the measurement matrix and $\bz$ is the noise vector with each entry being independently and identically distributed (i.i.d.) Gaussian variables with zero mean and variance $\sigma^2$.

Throughout the paper, we use bold capital letter and bold normal letter to denote a matrix and a vector, respectively. Given a matrix $\bA$, $A_{ij}$ denotes the entry located at the $i$-th row and $j$-th column, and $\ba_i$ denotes the $i$-th column. Given $i \in \{0, \cdots n-1 \}$, $(i)_2$ is the $\log n$-bit binary representation of $i$ with $0$ and $1$ mapped to $1$ and $-1$, respectively. For example, $n=3$, $(2)_2 = [1, -1, 1]$. Let $sgn(x) =1$ if $x \geq 0$ and $sgn(x)=-1$ otherwise.

\section{Measurement Matrix Design}

The LDPC inspired design of the measurement matrix is proposed in~\cite{pawar2012hybrid,bakshi2012sho}. In particular, the measurement matrix is constructed as
\begin{align}\label{eq:meas_matrix}
\bA = \bH \odot \bG
\end{align}
where $\bH \in \{0,1\}^{b \times n}$, $\bG \in \mathbb{R}^{c \times n}$ and the $\odot$ operator is defined as
\begin{align}
\bH \odot \bG  = \left[
\begin{array}{ccc}
H_{0,0} \bg_0 &  \cdots & H_{0,(n-1)} \bg_{n-1} \\
\vdots & \cdots & \vdots \\
H_{b-1,0} \bg_0 &  \cdots & H_{b-1,(n-1)} \bg_{n-1}
\end{array}
\right].
\end{align}
The number of measurements is thus $m = b \times c$.
For example,

\begin{align}
  \begin{bmatrix}
    1 & 0 & 1 \\
    0 & 1 & 1
  \end{bmatrix}
            \odot
            \begin{bmatrix}
              \bg_0 & \bg_1 & \bg_2
            \end{bmatrix}
                              =
                              \begin{bmatrix}
                                \bg_0 & 0 & \bg_2 \\
                                0 & \bg_1 & \bg_2
                              \end{bmatrix}.
\end{align}

In fact, $\bH$ is inspired by the parity-check matrix of LDPC codes. The relationship between the signal entries and the measurements can be represented by a bipartite graph. In the bipartite graph, there are $n$ left nodes with $x_i$ corresponding to the $i$-th left nodes, and $b$ right nodes, which are also referred to as bins. The $i$-th left node is connected with the $j$-th bin if $H_{ij} = 1$. The measurement vector is thus grouped into $b$ sub-vectors as $\by = \left[ \by_0^{\dagger}, \cdots, \by_{b-1}^{\dagger} \right]^{\dagger}$, where $\by_j \in \mathbb{R}^c$ is the $j$-th bin value given by
\begin{align}
\by_j = \sum_{i=0}^{n-1} H_{ij} x_i \bg_i + \bz_j.
\end{align}
Fig.~\ref{fig:bipartite_graph} illustrates the bipartite representation between signals and measurements. In this paper, we construct $\bH$ from the ensemble of left $d$-regular bipartite graph $\mathcal{G}_d(k,b)$, where every signal is connected to $d$ measurement bins uniformly at random.

\begin{figure}
  \centering
  \includegraphics[width=2in]{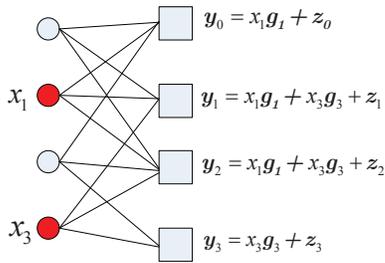}\\
  \caption{Example of the bipartite graph. Left nodes correspond to signals and right nodes correspond to measurement bins. The left nodes marked in red are nonzero signal components.}\label{fig:bipartite_graph}
\end{figure}

The recovery algorithms adopts the framework proposed in~\cite{li2014sub}. The recovery algorithm calls for a robust bin detection, which can 1) identify if a measurement bin is connected to no nonzero signal component (zeroton), to a single nonzero component (singleton) or to multiple nonzero components (multiton); 2) robustly estimate the signal index and value from singleton bins. It can be proved that by some proper $b = O(k)$, the recovery algorithm can correctly estimate $\bx$ with high probability if we have a robust bin detection. The key challenge is how to design $\bG$ to achieve robust bin detection.

In previous works~\cite{bakshi2012sho,pawar2012hybrid,li2014sub}, $\bG$ is constructed based on the DFT matrix. The signal index information $i$ is embedded in the phase difference between the entries of $\bg_i$. In this paper, we propose a new design of $\bG$, which only consists of $\{\pm 1\}$ entries and is more robust to noise and quantization errors.

We motivate the design using a simplified setting. Assume 1) a measurement bin $j$ is known to be a singleton, 2) there is no noise, and 3) the sign of the signal $x_i$ that is hashed to bin $j$ is known. The question is how can we design $\bG$ to detect the signal index $i$ and its value? Let $(i)_2$ be the $\log n$-bit binary representation of $i$. If $\bg_i = (i)_2$, then the signal index can be easily recovered based on the signs of each entry in $sgn(x_i) \by_j = | x_i | \bg_i$. A robust design of $\bG$ is to overcome the challenges posed by the three assumptions. First, we let $\bar{\bg}_i$ to be an all-one vector such that the signs of $x_i$ can be estimated. Second, $\tilde{\bg}$ is designed to be coded bits of $(i)_2$ for robust estimation of $(i)_2$ under the noisy setting. The sub-vector length is $\lceil \log n \rceil / R$, where $R$ is the code rate of the applied low-complexity error-control code \cite{barg2004error}. Third, we let $\dot{\bg}_i$ be a binary vector with each entry generated according to i.i.d. Rademacher distribution for singleton verification.

In all, the $i$-th column of $\bG$ consists of three sub-vectors:
\begin{align}
\bg_i = \left[\tilde{\bg}_i^{\dagger}, \bar{\bg}_i^{\dagger}, \dot{\bg}_i^{\dagger}  \right]^{\dagger}
\end{align}
where $\tilde{\bg}_i \in \{\pm 1 \}^{c_0
}$, $\bar{\bg}_i \in \{\pm 1 \}^{c_1
}$, and $\dot{\bg}_i \in \{\pm 1\}^{c_2
}$. Accordingly, the measurement vector for bin $j$ can be split into
three sub-vectors:
\begin{align}
\by_j &=
\left[
\begin{array}{cc}
 \tilde{\by}_j \\
 \bar{\by}_j \\
  \dot{\by}_j
\end{array}
\right]
=
\left[
\begin{array}{cc}
 \sum_{i=0}^{n-1} H_{ij} x_i  \tilde{\bg}_i \\
 \sum_{i=0}^{n-1} H_{ij} x_i  \bar{\bg}_i \\
 \sum_{i=0}^{n-1} H_{ij} x_i  \dot{\bg}_i
\end{array}
\right] +
\left[
\begin{array}{cc}
 \tilde{\bz}_j \\
 \bar{\bz}_j\\
 \dot{\bz}_j
\end{array}
\right] .
\end{align}
In our design, we choose $b = O(k)$, $c = O(\log n)$ and $c= O(\log^2 n)$ for signals with known discrete alphabet and arbitrary alphabet, respectively.

In the bipartite graph, each measurement bin can be regarded as a super check node where a subcode is further used to encode the index information of the signals. The structure is similar to that of generalized LDPC codes~\cite{miladinovic2008generalized}. The well-established low-complexity capacity-approaching point-to-point codes can serve as subcodes to enhance the robust design.

\section{Recovery Algorithm Design}
\label{sec:algo}

We adopt the recovery algorithm framework proposed in~\cite{pawar2012hybrid}. The algorithm is implemented in an iterative ``peeling" process. In every iteration, a singleton bin is identified. The index and value of the signal that is hashed to the singleton bin are estimated. The contribution of the estimated signal to the other connected bins are cancelled out (peeled off).

The main difference of our work lies in the signal support estimation
from a singleton bin, referred to as the {\em singleton test}, which is
described in Algorithm~\ref{algo:bin_detect}. In particular, for some $x_i$ that is hashed to a singleton bin $j$, $\tilde{\bg}_i \in \{ \pm 1\}^{c_0}$ encodes the support information $(i)_2$. Suppose the signal sign estimation is correct, i.e., $s = sgn(x_i)$. Without noise, $s \tilde{\by}_j = |x_i| \tilde{\bg}_i$ and thus $sgn(s \tilde{\by}_i)$ is exactly $\tilde{\bg}_i$. Under the noisy setting, some of the signs are flipped due to noise, which can be regarded as transmission over the binary symmetric channel (BSC). With low-complexity codes used as subcodes, $(i)_2$ can be recovered by inputting $sgn(s \tilde{\by}_i)$ to the corresponding decoder with complexity $O( c_0)$~\cite{barg2004error}.

The overall recovery algorithm is described as follows.

First, run the singleton test on every bin
 using Algorithm~\ref{algo:bin_detect}.
Let $L$ denote the set of estimated signal indices.  Remove the
  declared singleton bins.

Then, repeat the following
until $L = \emptyset$:
\begin{compactitem}
\item Select arbitrary 
  $i \in L$ and 
  remove $i$ from set $L$.
\item For every remaining 
  bin $j$ with $H_{ij} = 1$, perform the following:
  \begin{compactenum}
  \item Subtract the signal node $i$ value from bin $j$: $\by_j \leftarrow \by_j - \hat{x}_i \bg_i$.
  \item Run the singleton test on the bin
      using Algorithm~\ref{algo:bin_detect}.
      If it is a singleton, add the output index to
        $L$ and remove bin $j$. 
  \end{compactenum}
\end{compactitem}

Algorithm~\ref{algo:bin_detect} has a computational complexity of
$O(c)$, where $c = c_0 + c_1 + c_2$.
Performing the singleton test on all $b$ bin takes complexity
  $O(bc)$.  In each subsequent
iteration, we perform Algorithm~\ref{algo:bin_detect} only on every
(remaining)
connected
bin of a recovered signal component. Since the left-node degree is constant, each iteration involves computational complexity of $O(c)$. It will be proved that the algorithm terminates after $O(k)$ iterations with high probability. The computational complexity of all the iterations involved is thus $O(k c)$. With the choice of $b$ and $c$, the total complexity is $O( k \log n)$ and $O \left( k \log^2 n \right)$ for signals with discrete alphabet and arbitrary alphabet, respectively.
%
%

\begin{algorithm}
\caption{Singleton test} 
\label{algo:bin_detect}
\begin{algorithmic}[]
\STATE \textbf{Input}: Bin measurements, 
$\by = [\tilde{\by}^{\dagger} \enskip \bar{\by}^{\dagger} \enskip \dot{\by}^{\dagger}]^{\dagger}$.
\STATE \textbf{Output}: index and estimate, $(i,z_i)$.
\IF {$  || \dot{\by} ||^2 < c_2 (1+ \tau)  \sigma^2$}
\STATE Claim bin is a zeroton and return $(i,z_i) \leftarrow (\emptyset, 0)$.
\ENDIF
\STATE \textit{Signal sign estimation:}
$
s \leftarrow sgn (\bar{\bg}^{\dagger} \bar{\by}).
$
\STATE \textit{Signal index estimation:}
$
i  \leftarrow \text{BSC-Decoder} ( sgn( s \cdot \tilde{\by} ) ).
$
\STATE \textit{Singleton verification:}
\STATE $z' \leftarrow \frac{1}{c_2} \dot{\bg}_{i}^{\dagger} \dot{\by}$.
\IF {$|| \dot{\by} - z' \dot{\bg}_{i} ||_2^2 \leq (c_2-1) (1+\tau) \sigma^2  $ }
\STATE {Return $i$ and }
\begin{align*}
z_i & \leftarrow \left\{
\begin{array}{cc}
\frac{1}{c} \bg_i^{\dagger} \by & \text{ arbitrary alphabet}\\
\arg\min_{z' \in \mathcal{X}} ||\dot{\by} - z' \dot{g_{i}}||^2 & \text{discrete alphabet}
\end{array}
\right.
\end{align*}
\ELSE
\STATE {Claim bin is multiton and return $(i,z_i) \leftarrow (\emptyset,0)$.}
\ENDIF
\end{algorithmic}
\end{algorithm}

\section{Main Results and Proof}
\label{sec:proof}

\begin{theorem}\label{thm:cs_known}
  Given 
  any $\epsilon >0$, there exists $k_0>0$ 
  such that for every $k > k_0$ and 
  every $n$-dimensional $k$-sparse signal $\bx$ whose entries take
  their values in a known discrete alphabet, the proposed scheme
  achieves $\mathsf{P} \{\hat{\bx} \neq \bx \} < \epsilon$.
  The number of measurements required is $O(k\log n)$.  The
    computational complexity is $O(k\log n)$
    arithmetic operations.
\end{theorem}

\begin{theorem}\label{thm:cs_general}
  Given 
  any $\delta, \epsilon >0$, there exists $k_0>0$ 
  such that for every $k > k_0$ and 
  every $n$-dimensional $k$-sparse signal $\bx$ with $|x_i| \geq
  \delta$ for every 
  $i \in {\rm supp} (\bx)$, the proposed scheme achieves $\mathsf{P}
  \{ {\rm supp}(\hat{\bx}) \neq {\rm supp}(\bx) \}  < \epsilon$ and
  $\mathsf{P} \{ |\hat{x}_i - x_i|^2 \geq \epsilon \}  < \epsilon$
  for every 
  $i \in {\rm supp}(\bx) $.
  The number of measurements required is $O(k\log^2 n)$.  The
    computational complexity is $O(k\log^2 n)$
    arithmetic operations.
\end{theorem}

We focus on the proof of Theorem~\ref{thm:cs_general}
due to space limitations. Theorem~\ref{thm:cs_known} follows as a special case. Unlike signals from discrete
alphabet, the 
signal estimates have residual errors, which propagate to later iterations due to the peeling process. In this paper, we propose an \textit{error propagation graph} to keep track of the accumulated errors.

An error propagation graph for $x_i$ is a subgraph induced by the recovery algorithm, which contains the signal nodes that are estimated in the previous iterations and have paths to $x_i$. Fig.~\ref{fig:err_prop} illustrates the the error propagation graph for $x_2$.

\begin{figure}
  \centering
  \includegraphics[width=6cm]{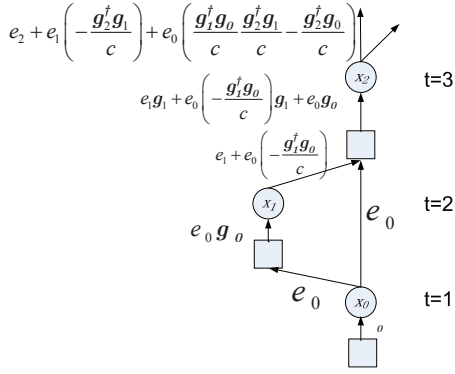}\\
  \caption{Error propagation graph for signal $x_2$.}\label{fig:err_prop}
\end{figure}

Define the estimation error of $x_i$ as
\begin{align}
p_i &= x_i - \hat{x}_i.
\end{align}
Let $m(i)$ be the measurement bin used to recover the signal index $i$. Define the point error of $x_i$ as
\begin{align}\label{eq:res_err}
e_i = - c^{-1} {\bg}_i^{\dagger} {\bz}_{m(i)} .
\end{align}
Then $e_i$ is a Gaussian variable with zero mean and variance $\sigma^2/ c$. We will keep track of $p_i$ using the error propagation graph.

Let $S(t)$ denote the signal indices that are recovered in the $t$-th iteration. Consider the estimation of $x_i$, $i \in S(1)$. The measurement vector of bin $m(i)$ and the residual estimation error are given by
\begin{align}
\by_{m(i)} &= x_i \bg_i + \bz_{m(i)} \\
p_i &= e_i.
\end{align}

Consider the estimation for $x_i$, $i \in S(2)$. With the peeling of $\hat{x}_\ell$, $\ell \in S(1)$, the updated measurement vector of $m(i)$ and the estimation error become
\begin{align}
\by_{m(i)} &= x_i \bg_i + \bz_{m(i)} + \sum_{\ell \in S(1): H_{\ell, {m(i)}} = 1} e_{\ell} \bg_{\ell} \\
p_i &= e_i +   \sum_{\ell \in S(1): H_{\ell, {m(i)}} = 1} e_{\ell} \left(- c^{-1} \bg_i^{\dagger} \bg_{\ell} \right),
\end{align}
where $|\bg_i^{\dagger} \bg_{\ell} /c| \leq 1$ for every realization of $\bG$.

The estimation error can be calculated recursively according to some message passing rules over the graph. In particular, let $p_i$ be the estimation error propagated from signal node $i$ and $\bq_j \in \mathbb{R}^{c}$ be the error vector propagated from the measurement bin $j$. The errors can be calculated according to the following rules:
\begin{align}
\label{eq:err_prop_rule1} p_i &= e_i + \sum_{j \in \text{in}(i) } \left( - c^{-1} \bg_i^{\dagger} \bq_j \right) \\
\label{eq:err_prop_rule2} \bq_j &= \sum_{i \in \text{in}(j)} p_i \bg_i,
\end{align}
where $\text{in}(i)$ denotes the indices of the measurement bins (signal nodes) incoming to signal node (measurement bin) $i$.

By induction and the error message passing rules~\eqref{eq:err_prop_rule1} and \eqref{eq:err_prop_rule2}, the error propagation effect is characterized by the following lemma.
\begin{lemma}
The estimation error of $x_i$, $i \in S(t)$, is calculated as
\begin{align}
p_i & = e_i +  \sum_{\ell \in \cup_{j =1}^{t-1} S(j) \cap D(i)} \sum_{p = 1}^{P(\ell, i) } e_{\ell} d_{\ell,p},
\end{align}
where $D(i)$ be the connected subgraph of the bipartite graph containing $i$, $\mathcal{P} (\ell, i)$ is the number of paths from $\ell$ to $i$ in $D(i)$, and $d_{\ell, p}$ is some coefficient depending on both $\bG$ and the path satisfying $|d_{\ell,p}| \leq 1$.
\end{lemma}

Fig.~\ref{fig:err_prop} illustrates an example. The number of paths
from $x_0$ to $x_2$ is $\mathcal{P}(0,2) = 2$, with the corresponding
coefficients being $d_{0,1} = - 
c^{-1} \bg_2^{\dagger} \bg_0 
$ and $d_{0,2} =  c^{-2} 
\bg_1^{\dagger} \bg_0 
\bg_2^{\dagger} \bg_1 
$. The number of paths from $x_1$ to $x_2$ is $\mathcal{P}(1,2) = 1$,
with the coefficients being $d_{1,1} = - c^{-1} 
\bg_2^{\dagger} \bg_1
$.

We further bound the errors by leveraging results on random hypergraph. The bipartite graph induced by $\bH$ corresponds to a hypergraph where the left nodes and right nodes represent hyperedges and vertices. The hyperedge $i$ is incident on vertex $j$ if $H_{ij} = 1$. Then the random bipartite graph $\mathcal{G}_d (k,b)$ induces a $d$-uniform random hypergraph.
\begin{lemma}\cite{karonski2002phase}\label{lemma:hypergraph_component}
Suppose $b/k$ is some constant large enough, then with probability $1- O(1/k)$, $\mathcal{G}_d (k,b)$ contains only trees or unicyclic components, and the largest component contains $O( \log k)$ signal nodes.
\end{lemma}

Let $\mathcal{E}_H$ denote the event that the bipartite graph satisfies the condition as described in Lemma~\ref{lemma:hypergraph_component} with $P \{\mathcal{E}_H\}  =1- O(1/k)$. We first bound the detection error probability conditioned on $\mathcal{E}_H$. Suppose $\mathcal{E}_H$ holds, then $\mathcal{P}(i,j) \leq 2$, otherwise the component is not unicyclic. Moreover the largest component $D(j)$ contains $O(\log k)$ signal nodes. Therefore, conditioned on $\bG$, $p_i$ is Gaussian distributed with zero mean and the variance of can be upper bounded as
\begin{align}
\label{eq:var_ub} {\rm var} (p_i) &\leq \left( 1+ \sum_{\ell \in D(i)} P^2(\ell,j) \right) \frac{\sigma^2}{c} = O\left( \log k \frac{\sigma^2}{c} \right).
\end{align}

\begin{lemma}\label{lemma:singleton_estimate}
Conditioned on that $\mathcal{E}_H$ holds, given any $\delta >0$ and $|x_i| \geq \delta$, $\forall i \in {\rm supp} (\bx)$, the recovery algorithm can correctly identity the signal support with probability $1 - O(1 / n)$ with some $c = O(\log^2 n)$.
\end{lemma}
\begin{proof}
The support detection may be subject to zeroton, multiton and singleton detection errors. The error probability of detecting zerotons and multitons can be upper bounded by $O(1/n)$ following similar steps in~\cite{li2014sub} and is omitted due to space limitation. We focus on the singleton detection.

Suppose the measurement vector of a singleton is given by
\begin{align}
\by_j = x_i \bg_i + \bw_j,
\end{align}
where the entries in $\bw_j$ are i.i.d. Gaussian variable with zero mean and variance $\tilde{\sigma}^2$. Then the error probability of sign estimation is calculated as
\begin{align}~\label{eq:err_sign}
\mathsf{P} \{sgn (\mathbf{1}^{\dagger} \bar{\by}) \neq sgn(x_i) \} &= \mathsf{P} \{ \mathbf{1}^{\dagger}  \bar{\bw} / c_1 \geq |x_i| \} \\
&= Q( \sqrt{c_1} |x_i| / \tilde{\sigma}),
\end{align}
where $Q(x)$ is the $Q$-function for standard normal distribution.

Suppose the signs of $x_i$ is correctly detected and we want to detect $(i)_2$ by recovering the signs of $\tilde{\bg}_i$. By compensating the signs of $x_i$ as $ sgn(x_i) \by_j$, the random transformation
\begin{align}
sgn (|x_i| \tilde{\bg}_i ) \to sgn( |x_i| \tilde{\bg}_i + \tilde{\bw}_j )
\end{align}
is equivalent to transmission over a BSC with crossover probability less than $Q(|x_i| / \tilde{\sigma})$~\cite{chen2015robust}.

From the recovery process, $\bw_j$ is the noise plus residual estimate errors given by $\bw_j = \bz_j + \sum_{\ell \in \text{in}(j)} p_{\ell} \bg_{\ell}$. According to \eqref{eq:var_ub}, for some $c = O \left( \log^2 n \right)$ and a large enough $n$, the variance of $\bw_j$ is dominated by that of $\bz_j$. The entries of $\bw_j$ have a variance $\tilde{\sigma}^2$ bounded by some constant. Therefore, given that $|x_i| \geq \delta$ for some $\delta$, the worst-case SNR for every singleton estimation is lower bounded by some constant. The error probability of sign estimation~\eqref{eq:err_sign} is $O(1/n^3)$ with some $c_1 = O(\log n)$. Applying an error control code of length $c_0 =  \log n /R$ with a low enough code rate $R$, $(i)_2$ can be decoded correctly with probability $1- O(1/n^3)$.

Note that if $\mathcal{E}_H$ holds and the singleton, multiton and zeroton bins are correctly estimated, the peeling decoder terminates by recovering every nonzero signal entry~\cite{price2011efficient}. Since there are at most $O(k)$ iterations and every iteration involves at most $O(k)$ singletons, the error probability can be upper bounded by $O(k^2 / n^3) = O(1/n)$ using the union bound. Moreover, conditioned on that $i$ is correctly estimated, the probability that the singleton verification is not passed is equivalent to a zeroton detection error, which can be upper bounded by $O(1/n)$. The lemma is hence established.
\end{proof}

Support recovery fails only if either $\mathcal{E}_H$ does not hold or that a bin detection error occurs conditioned on $\mathcal{E}_H$ holds. By Lemma~\ref{lemma:hypergraph_component} and Lemma~\ref{lemma:singleton_estimate}, the overall error probability of support recovery is $O \left( \frac{1}{n} + \frac{1}{k} \right)$, vanishing as $k$ increases.

The error probability can be upper bounded as
\begin{align}
\mathsf{P} \left\{ |x_i - \hat{x_i}|^2 \geq \epsilon \right\} & \leq  \mathsf{P} \left\{ \mathcal{E}_H^c \right\} + \mathsf{P} \left\{ | p_i | \geq \sqrt{\epsilon} | \mathcal{E}_H \right\} \\
\label{eq:mse} & \leq \mathsf{P} \left\{ \mathcal{E}_H^c \right\} + 2
                 Q \left( \sqrt{ \frac{\epsilon c}{\log k \sigma^2} }
                 \right) 
\end{align}
where \eqref{eq:mse} follows because $p_i$ is Gaussian variable with zero mean and variance upper bounded by \eqref{eq:var_ub} conditioned on $\mathcal{E}_H$ and every realization of $\bG$. By Lemma~\ref{lemma:hypergraph_component}, the error probability \eqref{eq:mse} is smaller than any $\epsilon$ with a large enough $k$ and some $c = O \left( \log^2 n \right)$. Hence, Theorem~\ref{thm:cs_general} is established.

\section{Simulation}

Throughout the simulation, we assume that the nonzero signal amplitude is taken uniformly at random from $[1, 10]$ and define SNR = $1/\sigma^2$, which is the worst-case SNR. The signal dimension is $n = 10^{10}$. The number of measurement bins is chosen to be $b = 3 k$.  We adopt a regular random LDPC code with rate $1/2$ as subcode to encode the signal support information, and thus $c_0 = 2 \log n$. We let $c_1 = \log n$ and $c_2 = 2 \log n$. Fig.~\ref{fig:errprob} and Fig.~\ref{fig:rmse} plot the error probability of support recovery and relative mean square error, respectively. The relative mean square error is only calculated and averaged over the signals with their support correctly estimated. We run 200 simulations for each SNR. In the simulation, for every sparsity level $k$, the error-control code and nonzero signal entries are generated once and fixed.

Although analysis shows that $c = O \left( \log^2 n \right)$ is sufficient to guarantee vanishing error probability, choosing $c = O(\log n)$ also gives a good performance. The error probability of support recovery and relative mean square error decreases as SNR increases. In order to achieve more reliable signal recovery, we can adopt a more sophisticated error-control code or a code with lower code rate.

\begin{figure}
  \centering
  \includegraphics[width=6cm]{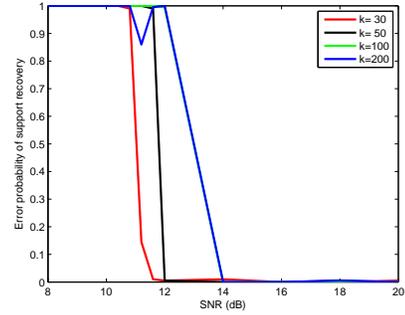}\\
  \caption{Error probability of support recovery.}\label{fig:errprob}
\end{figure}

\begin{figure}
  \centering
  \includegraphics[width=6cm]{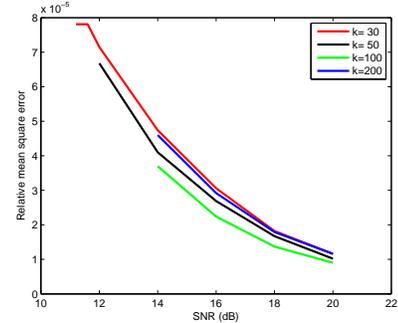}\\
  \caption{Relative mean square error.}\label{fig:rmse}
\end{figure}

\bibliographystyle{IEEEtran}
\bibliography{IEEEabrv,all_bib,xu_bib}





\end{document}